\documentclass[sigconf,9pt]{acmart}
\renewcommand\footnotetextcopyrightpermission[1]{} % removes footnote with conference information in first column
\AtBeginDocument{%
  \providecommand\BibTeX{{%
    \normalfont B\kern-0.5em{\scshape i\kern-0.25em b}\kern-0.8em\TeX}}}

\acmYear{2022}\copyrightyear{2022}
\setcopyright{acmcopyright}
\acmConference[e-Energy '22]{The Thirteenth ACM International Conference on Future Energy Systems}{June 28--July 1, 2022}{Virtual Event, USA}
\acmBooktitle{The Thirteenth ACM International Conference on Future Energy Systems (e-Energy '22), June 28--July 1, 2022, Virtual Event, USA}
\acmPrice{15.00}
\acmDOI{10.1145/3538637.3538867}
\acmISBN{978-1-4503-9397-3/22/06}

\usepackage[framemethod=1]{mdframed}

\usepackage{color}
\definecolor{shadecolor}{rgb}{0.878906, 0.878906, 0.878906}
\usepackage{multirow}
\usepackage{float}
\usepackage{framed}
\usepackage{amsmath}
\usepackage{amsthm}
\usepackage{graphicx}
\usepackage{tabularx}%\usepackage[square,numbers,sort]{natbib}

\usepackage{pifont}% http://ctan.org/pkg/pifont

\makeatletter

\floatstyle{ruled}
\newfloat{algorithm}{tbp}{loa}
\providecommand{\algorithmname}{Algorithm}
\floatname{algorithm}{\protect\algorithmname}

\theoremstyle{plain}
\newtheorem{thm}{\protect\theoremname}
\theoremstyle{definition}
\newtheorem{defn}{\protect\definitionname}
\theoremstyle{plain}
\newtheorem{lem}{\protect\lemmaname}
\theoremstyle{plain}

\usepackage{graphics, subfig}
\usepackage{epsfig}
\@ifundefined{definecolor}
 {\usepackage{color}}{}
\usepackage{amsfonts}
\usepackage{latexsym}
\usepackage{tabularx}%\usepackage[square,numbers,sort]{natbib}
\usepackage{dsfont}
\usepackage{comment}
\usepackage{colortbl}
\usepackage{romannum}
\usepackage{mathtools}
\usepackage{algpseudocode}
\algnewcommand\algorithmicforeach{\textbf{for each}}
\algdef{S}[FOR]{ForEach}[1]{\algorithmicforeach\ #1\ \algorithmicdo}
\algnewcommand{\AND}{\algorithmicand}

\ifodd1%revise of the text
\newcommand{\com}[1]{\textbf{\color{blue} (COMMENT: #1)}}%comment of the text

\else\newcommand{\com}[1]{}\fi

\newcommand{\beq}{\begin{equation}}\newcommand{\eeq}{\end{equation}}\newcommand{\bea}{\begin{eqnarray}}\newcommand{\eea}{\end{eqnarray}}\newcommand{\bda}{\begin{eqnarray*}}\newcommand{\eda}{\end{eqnarray*}}\newcommand{\bdalign}{\begin{align*}}\newcommand{\edalign}{\end{align*}}

\DeclareMathOperator{\EX}{\mathbb{E}}% expected value

\@ifundefined{showcaptionsetup}{}{%
 \PassOptionsToPackage{caption=false}{subfig}}
\usepackage{subfig}
\makeatother

\usepackage[english]{babel}
\providecommand{\corollaryname}{Corollary}
\providecommand{\definitionname}{Definition}
\providecommand{\lemmaname}{Lemma}
\providecommand{\theoremname}{Theorem}

\newcommand{\subsubsectionCUSTOM}[1]{\subsubsection{\textbf{#1}}{~}\\}

\usepackage{lipsum}

\usepackage{nopageno}

\begin{document}

\title{Optimization of Cryptocurrency Miners' Participation in Ancillary Service Markets}
\thispagestyle{empty}
\pagestyle{plain}
\pagestyle{empty}
\thispagestyle{plain}

\author{Ali Menati}
\affiliation{
  \institution{Department of Electrical Engineering}
   \country{Texas A\&M University}}

\author{Yuting Cai}
\affiliation{
  \institution{Department of Electrical Engineering}
   \country{Texas A\&M University}}

\author{Rayan El Helou}
\affiliation{
  \institution{Department of Electrical Engineering}
   \country{Texas A\&M University}}

   \author{Chao Tian}
\affiliation{
  \institution{Department of Electrical Engineering}
   \country{Texas A\&M University}}

\author{Le Xie}
\affiliation{
  \institution{Department of Electrical Engineering}
   \country{Texas A\&M University}}

\authornote{Corresponding  author: Le Xie.}

\begin{abstract}
Proof-of-work computation used in cryptocurrencies has witnessed significant growth in the U.S. and many other regions around the world. One of the most significant bottlenecks for the scalable deployment of such computation is its energy demand. On the other hand, the electric energy system is increasing the need for flexibility for energy balancing and ancillary services due to the intermittent nature of many new energy resources such as wind and solar. In this work, we model the operation of a cryptomining facility with heterogeneous mining devices participating in ancillary services. We propose a general formulation for the cryptominers to maximize their profit by strategically participating in ancillary services and controlling the loss of mining revenue, which requires taking into account the disparity in the efficiency of the mining machines. The optimization formulation is considered for both offline and online scenarios, and optimal algorithms are proposed to solve these problems. As a special case of our problem, we investigate cryptominers' participation in frequency regulation, where the miners benefit from their fast-responding devices and contribute to grid stability. In the second special setting, a risk-aware algorithm is proposed to jointly minimize the cost and the risk of participating in ancillary services with homogeneous mining devices. Simulation results based on real-world Electric Reliability Council of Texas (ERCOT) traces show more than $20\%$ gain in profit, highlighting the advantage of our proposed algorithms.
\end{abstract}

\pagestyle{empty}
\thispagestyle{plain}
\settopmatter{printfolios=true}
\maketitle 
\pagenumbering{gobble}

%%%%%%%%%%%%%%%%%%%%%%%%%%%%%%%%%%%%%%%
% REMOVE PAGE NUMBERS BEFORE SUBMISSION
%%%%%%%%%%%%%%%%%%%%%%%%%%%%%%%%%%%%%%%
%%%%%%%%%%%%%%%%%%%%%%%%%%%%%%%%%%%%%%%
\section{Introduction}
Cryptomining is emerging as one of the fastest-growing consumers of electric power in many regions, such as Texas and Georgia in the U.S., creating significant challenges for electric grid operators, both in planning and in real-time operations. According to the report from the White House, cryptomining currently consumes $0.9\%$ to $1.7\%$ of total electricity in the U.S., and it is still rapidly growing \cite{Whitehouse2022}. At the same time, there is a growing need for flexibility in the electric grid to provide ancillary services due to the intermittent nature of renewable energy sources.

Ancillary services are one of the most critical mechanisms for providing the flexibility needed for the large-scale integration of variable energy resources \cite {wu2021open,el2022impact}. Nevertheless, an inconvenient truth is that most energy consumption demands tend to be positively correlated with the load curve of the grid, i.e., increasing power demand when the grid is under stress. %Generally, most demand is also at its high points. 
Thus, the flexibility that can be extracted from these demands is limited due to the primary need for serving customers. 

Unlike most other demands, cryptomining demand can have much stronger flexibility during times of need when the grid is stressed. As shown in Figure 1, 
the actual cryptomining load in the ERCOT region has a significant negative correlation (-0.757) with the system-wide net load during the summer peak time  \footnote{Summer peak time refers to the period of July 7th, 2022 to July 21st, 2022 \cite{menati2022carbon}.}. This is in sharp contrast with most other loads, including conventional Internet data centers. Such negative correlation suggests a great potential for cryptomining facilities to offer aggresive ancillary services in electricity markets. 

\begin{figure}[!tbp]
    \centering
    \includegraphics[width =\columnwidth, trim={0 5em 0 0em}]{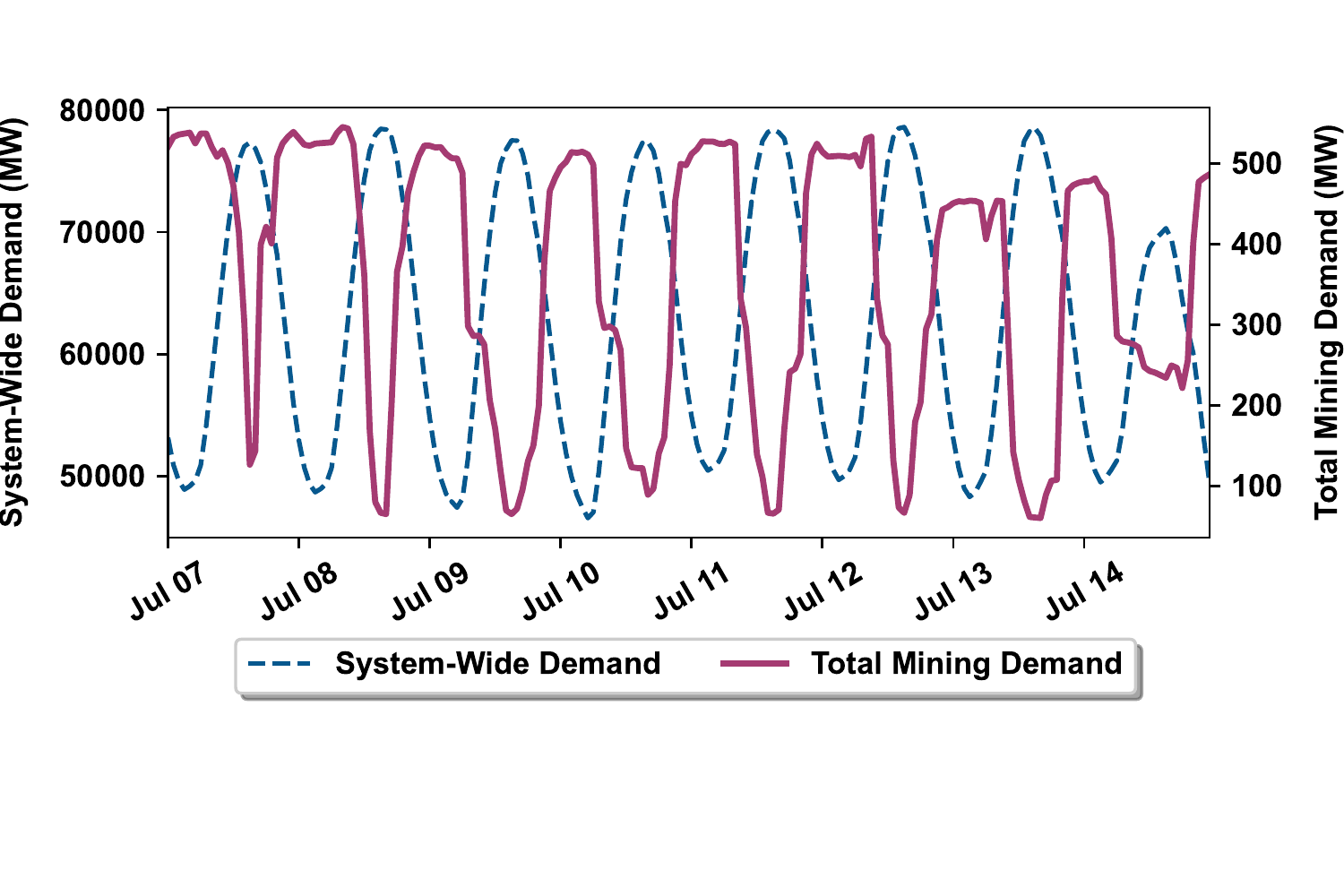}
    \caption{ERCOT system-wide demand vs. total cryptomining demand for a week in summer 2022.}
    \label{fig:carbon}
\end{figure}

Participating in ancillary services has not been wildly adopted by cryptomining facilities, given the sophisticated mechanisms and technology set up. However, such participation could be a substantial source of revenue for mining facilities. For example, during December 2022, a mining facility in Texas earned over \$4.9 million in demand response credits, constituting 30\% of its total revenue \cite{riot}. Many cryptominers have expressed strong interests in participating in ancillary services. At the same time, grid operators are also anticipating such participation, and have a strong need to better understand the behavior of cryptominers when they participate. In this work, we propose a general formulation for the cryptominers to optimize their operation with ancillary service participation, and  control their loss of mining revenue. We
summarize our main contributions as follows:

(1) Novel optimization problem formulations are introduced to allow cryptomining facilities to participate aggressively in the ancillary services market. The formulation is general for facilities with multiple devices and multiple ancillary service programs available. The setting of the optimization problem includes both online and offline versions. Two special cases are also considered: The first case has two dependent programs with joint deployment probability distribution; the second case is for a cryptomining facility with a risk-aware mechanism and single-type devices.

(2) Efficient algorithms are proposed to solve the optimization problems for cryptomining facility participation with potentially heterogeneous cryptomining devices. The proposed algorithms are guaranteed to converge to the optimal solutions in both the offline and online settings, in terms of expected cost and in terms of the regret, respectively. 

(3) Extensive numerical simulations are performed based on real-world ERCOT data for both the special cases discussed earlier and for the general setting. These simulations show that with the proposed algorithms, cryptomining facilities can sometimes expect improved profit margins of more than $20\%$ by strategically participating in the ancillary services. Moreover, simulation in the online setting shows that the proposed algorithms are effective in this setting by quickly approaching the best strategy in hindsight. 

The remainder of this paper is organized as follows. Related work is presented in Section \ref{sec:background}. In Section \ref{sec:problem}, a general formulation that allows cryptomining to participate in ancillary services is presented. Section \ref{sec:algorithm} provides the proposed algorithms. Section \ref{sec:dependent programs} investigates the optimal solution for dependent frequency regulation programs specifically. Section \ref{sec:risk_aware_formulation} studies the risk-aware setting. Section \ref{sec:online} considers the online version of the general problem. Numerical experiments based on real-world case studies are presented in Section \ref{sec:experiment}, and Section \ref{sec:conclusion} concludes the paper.
\section{Background and related work}\label{sec:background}
Over the past few years, with the substantial increase in the integration of renewables and electrification of the residential and transportation sections, uncertainty and intermittency have increased on both the load and generation sides. To address these challenges, various demand response and ancillary service programs have been extensively studied by researchers \cite{paterakis2017overview, hatami2010minimizing, jiang2011real}. However, considering the large number of customers and their uncertain behavior, optimizing the decision-making process for ancillary services is still a challenging task \cite{tsui2012demand, ye2016game, ibars2010distributed}.

Participating in ancillary services is a win-win strategy, which helps grid achieve better load balancing, and helps data centers reduce their energy bills \cite{aikema2012data}, and create a new source of revenue \cite{chen2013dynamic}. Data center participation in ancillary services is closely tied with their optimal workload management to minimize their operational costs \cite{Rao12, Ghamkhari13, Chen16}. In literature, data center workload scheduling is usually combined with local power generation \cite{liu2013data}, or cooling systems \cite{Cupelli18}. Some other works investigate data center demand response from the utilities' point of view by designing optimal pricing mechanisms to encourage data centers' participation in ancillary services \cite{Wang16, liu2011greening, Tran16, Bahrami19}. 

Frequency regulation is an integral part of ancillary services, which ensures uninterrupted supply-demand balance and avoids large frequency deviations. Traditional methods of providing frequency regulation, such as thermal fuel and hydro generators, can be expensive, and more importantly, they are rather slow \cite{hao2015potentials}. Recent studies have found that large flexible loads like cryptocurrency mining facilities can provide significantly faster response \cite{lancium}. While there have been some studies on data centers participation in frequency regulation, due to the limited flexibility of data center \cite{quirk2021cryptocurrency}, most of these studies require pairing the data center with battery storage systems \cite{shi2017using, aksanli2014providing, li2014integrated}.

There are some similarities between the problem of cryptocurrency mining facilities participating in demand response and that of data center participation \cite{wierman2014opportunities,
Cupelli18}. However, the inherent structure of the two problems is different due to the fundamental differences in their physical and economic characteristics. First, modeling data center operation usually requires considering a start-up cost associated with the damage caused by turning devices on and off. However, in mining facilities, this is usually negligible, since the mining device are taken out of operation, and frequently being updated with newer, more efficient machines. For example, between August 2016 and July 2022, the average estimated energy intensity of mining machines decreased by around 85\%, due to technological improvements and the older machines being replaced with the newer ones \cite{WhiteHouse_crypto_report}. 
Second, the data center's demand flexibility is highly dependent on its processing requests and changes over time \cite{liu2013data}. However, mining facility is considered as \emph{complete} flexible demand capable of participation in demand response with full capacity \cite{menati2022modeling}. Third, while data centers usually use homogeneous servers with similar efficiency, in mining facilities we often observe multiple types of mining devices, each with a different mining capability. Later in Section \ref{sec:problem}, we will see the impact of heterogeneous devices on the problem formulation and the optimal choice of ancillary service profile. To the best of our knowledge, this is the first work to consider the optimization of cryptocurrency miners' participation in ancillary service markets.

\section{Ancillary Service Participation Problem}
\label{sec:problem}

The objective of the ancillary service profile selection by cryptocurrency mining facilities is to utilize the demand flexibility of the cryptomining devices to maximize cryptominer's profit. At the same time, this active participation will help the grid achieve better supply-demand balance. We consider a system that operates in a time-slotted fashion, where $\mathcal{T}:= \{1, ..., T\}$ denotes the set of all time slots, and $T$ denotes the total length of the time horizon ($T\triangleq |\mathcal{T}|$). The key notations are presented in Table~\ref{tbl:not}.

\subsection{System Model}

Let us consider a cryptomining facility with $K$ different types of cryptomining devices, where $\mathcal{K}$ represents the set of all possible cryptomining devices. We denote $c_k^M$ as the total capacity of type $k$ machines, where the aggregated capacity of the cryptomining facilities is $C^M=\sum_{k=1}^{K} c_k^M$. While cryptocurrency mining is the main source of revenue for the facilities, they also have the capability to participate in different ancillary service programs. We consider $\mathcal{N}:=\{1, ..., N\}$ as the set of all available ancillary service programs that the cryptomining facility is qualified to participate in, and $c_i(t)$ as the share of the facility capacity committed to program $i$ at time $t$, where $i \in \mathcal{N}$. Here $c_i(t) \leq \sum_{k=1}^{K} c_k^M=C^M$, and $c(t) =\big[c_i(t)\big]_{i=1}^{N}$ is an $N$ dimensional vector denoting the ancillary service profile at time $t$. Depending on time, necessity, and physical characteristics of ancillary service programs, cryptominers receive different revenues for participating in each program. Let us denote $p_i(t)$ as the expected per-unit price of the $i$th ancillary service program at time $t$, and $\epsilon_i(t)$ as the percentage of the committed capacity $c_i(t)$ that is being deployed at time $t$, where $\epsilon_i(t)c_i(t)$ capture the amount of deployed capacity under ancillary service program $i$. Here, $\epsilon(t) =\big[\epsilon_i(t)\big]_{i=1}^{N}$ denotes the deployment rate of all ancillary service programs at time $t$, where $\epsilon_i(t)$ is a random variable between zero and one that is highly dependent on the type of program. 

\begin{table}[!t]
	\begin{center}
		\begin{tabular}{|c|c|p{6cm}|}
			\hline
			\multicolumn{2}{|c|}{\textbf{Notation}} & \textbf{Definition} \\
			\hline \hline
			\multirow{6}{*}{\rotatebox[origin=c]{90}{\hspace{0mm} \textbf{Programs} }} 	
			&$p_i(t)$  &  Per-unit price of the $i$th ancillary service program at time $t$ (\$/MWh)\tabularnewline
			&$c_{i}(t)$  & Committed capacity to program i at time $t$ (MW)\tabularnewline
			&$\epsilon_{i}(t)$  & Deployment rate of the $i$th program at time $t$ \tabularnewline
			&$\mathcal{N}$  & The set of possible ancillary service programs\tabularnewline 
            &$\mathcal{T}$  & The set of time slots ($T\triangleq |\mathcal{T}|$)\tabularnewline \hline
			\multirow{7}{*}{\rotatebox[origin=c]{90}{ \hspace{-30mm} \textbf{Cryptominer}}}
            &$\mathcal{K}$ & The set of all possible types of cryptomining devices
			\tabularnewline
			&$c_k^M$  & Maximum capacity of the  $k$th type of cryptomining devices (MW)
            \tabularnewline
			&$C^M$  & Total aggregated capacity of all cryptomining devices (MW)\tabularnewline
			&$r_{k}(t)$ & The per unit net reward of mining cryptocurrency with the $k$th type of machines at time $t$ for one unit of electricity (\$/MWh) \tabularnewline
            &$d_{k}(t)$ & The deployed capacity covered using the $k$th type of machines at time $t$ (MW) \tabularnewline
            &$p_b^k(t)$ & The revenue of the $k$th type of machines obtained from selling the cryptocurrency mined using one unit of electricity at time $t$ (\$/MWh) \tabularnewline
            &$p_e(t)$ & The per-unit electricity cost (\$/MWh) \tabularnewline
            \hline
		\end{tabular} %\vspace{-0.8\baselineskip}
	\end{center}
	\caption{Key Notations.}
	\label{tbl:not}
	\vspace{-5mm}
\end{table}

We define $r_k(t)$ as the per unit net reward of mining cryptocurrency with type $k$ machines at time $t$ for one unit of electricity, and it is calculated as $r_k(t)=p_b^k(t)-p_e(t)$. $p_b^k(t)$ is the revenue of the $k$th type of machines obtained from selling the cryptocurrency mined using one unit of electricity, and $p_e(t)$ is the per-unit electricity cost. We define $p_b^k(t)$ as follows: 
\begin{equation}
p_b^k(t)=\frac{\textit{Crptocurrency price}  \,(\$/crypto)}{\textit{Energy intensity of type k machines} \,(MWh/crypto)}, \notag
\end{equation}
which depends on the cryptocurrency unit price and the energy intensity of the cryptomining machines. Energy intensity captures the amount of energy needed to mine one unit of cryptocurrency. For example, consider a Bitcoin mining facility, which needs $100MWh$ of energy for its $k$th type of machine to mine one bitcoin. If at time $t$, the Bitcoin price is $\$20,000$, we calculate $p_b^k(t)$ as $p_b^k(t)=\frac{20,000 \, \$/BTC}{100 \,MWh/BTC} \approx 200 \$/MWh$. 

If at time $t$, the cryptominer is participating in program $i$ with capacity $c_i(t)$, then its total committed capacity is $\sum_{i=1}^{N} c_i(t)$, and its total deployed capacity is $\sum_{i=1}^{N} \epsilon_i(t) c_i(t)$. 
The cryptominer needs to reduce its consumption and cover this deployed capacity by turning off a combination of its cryptomining devices. Let us denote $d_k(t)$ as the deployed capacity covered using the $k$th type of machine. Note that the $k$th type of machines are assumed to always mine with full capacity $c_k^M$, unless deployed by the cryptominer to reduce their consumption by $d_k(t)$ and operate at the new level of $c_k^M- d_k(t)$. In this paper, we assume that at each time slot $t$, $r_k(t) \geq 0$ for all $k \in \mathcal{K}$, which means that the reward of mining cryptocurrency is positive, and the cryptomining facility keeps mining, with full capacity if it is not deployed by ancillary services. We also assume that $p_i(t) \geq 0$ for all $i \in \mathcal{N}$, which always holds because the price of ancillary services is always positive. In the following section, we first determine the optimal deployment strategy for different types of cryptomining devices and then find the optimal ancillary service profile for the cryptominer. 
\subsection{Problem Formulation}
We consider a setting where the cryptominer chooses an ancillary service profile $c(t)$ at each time slot. After submitting its decision,  the cryptominer receives a signal from the grid operator to deploy a certain fraction ($\epsilon(t)$) of its committed capacity. Then, the cryptominer decides which types of cryptomining devices to shut down in order to fulfill its deployed capacity. Intuitively, it should start by turning off the more energy-intensive devices and keeping the more efficient ones running and mining. We formulate this decision-making process as follows:
\begin{subequations} \label{op:deployment}
	\begin{eqnarray}
	&{\min}& \sum_{t=1}^{T} \big[ \sum_{k=1}^{K} r_k(t) d_k(t) - \sum_{i=1}^{N}c_i
(t)p_i(t) \big]\\
	&\mbox{s.t.}\;&  \sum_{k=1}^{K} d_k(t) = \sum_{i=1}^{N} \epsilon_i(t) c_i(t),\quad \quad  \, \, \, \quad \textit{for} \, \, \, \,  t \in \mathcal{T} \label{op:deployment-constraint1}\\
&\mbox{vars.}\;& d_k(t) \in [0, c_k^M],   \quad \quad  \textit{for} \, \, \,   t \in \mathcal{T},  \, \, \,   and\, \, \, k \in \mathcal{K}, \label{op:deployment-constraint2}
 \end{eqnarray} 
\end{subequations}
where $r_k(t) d_k(t) $ captures the loss of revenue induced by stopping mining cryptocurrency due to the deployment of the $k$th cryptomining devices, and $c_i(t) p_i(t)$ is the revenue of participating in program $i$. Here  $\sum_{i=1}^{N} \epsilon_i(t)c_i(t)$ is the total capacity deployed under all ancillary service programs, and constraint \eqref{op:deployment-constraint1} ensures that this is covered by a combination of different cryptomining devices. Constraint \eqref{op:deployment-constraint2}  ensures that the deployed capacity covered by each type of cryptomining device is not more than its available capacity. 

It should be noted that all operational ASIC miners are fast-responding machines, capable of increasing or decreasing their power consumption within seconds \cite{lancium}. Hence, we do not consider any ramping constraint for these machines, and their deployment at each time slot is independent of other time slots. In other words, we can solve the optimization problem in~\eqref{op:deployment} separately for each time slot $t \in \mathcal{T}$. In the rest of this paper, we drop the time index $t$ for all the parameters and focus on the optimization over one time slot as follows:
\begin{subequations} \label{op:deploy}
	\begin{eqnarray}
	&{\min}&  \sum_{k=1}^{K} r_k d_k - \sum_{i=1}^{N}c_i
p_i\\
	&\mbox{s.t.}\;&  \sum_{k=1}^{K} d_k = \sum_{i=1}^{N} \epsilon_i c_i, \\
&\mbox{vars.}\;& d_k \in [0, c_k^M],   \quad \quad \quad \textit{for} \quad  k \in \mathcal{K},
 \end{eqnarray} 
\end{subequations}
In this optimization problem, the cost function is linear in terms of the optimization variable $d_k$. The optimal decision is different depending on the value of $r_k$. And the larger $r_k$, the smaller $d_k$ should be. This optimal solution implies that cryptomining devices with larger per-unit reward should be deployed the least frequently. With this intuition, in order to distribute the total deployed capacity $\sum_{i=1}^{N} \epsilon_i c_i $ among all types of machines, we start with the least efficient machines and only deploy the next efficient machines when the previous ones are completely deployed. Let us assume that the $K$ types of machines are sorted based on their per-unit reward in an ascending order $r_1 < r_2 < ...< r_K$. It should be noted that this order will not change through time, and the ascending order of cryptomining devices could be easily calculated and fixed. The following theorem determines the optimal choice of $d_k$ for all types of machines.

\begin{thm} For a cryptominer with the total deployed capacity of $\sum_{i=1}^{N} \epsilon_i c_i$, and $K$ types of cryptomining devices with $r_1 < r_2 < ...< r_K$, the optimal deployment by each type is:
\begin{equation} \label{deployed-calc}
		d_k \mbox{=} \begin{cases}
		 \min \{\sum_{i=1}^{N} \epsilon_i c_i, \,  c_1^M \}	\, & \mathrm{if } \quad  k = 1,\\
  \min \{\sum_{i=1}^{N} \epsilon_i c_i - \sum_{j=1}^{k-1} d_j, \,  c_k^M \}	\, & \mathrm{if } \quad  k \in [2:K].
		\end{cases} 
\end{equation} 
\label{thm:deployment}  
 \end{thm} 
\begin{proof}
Let us start with the optimal deployment of the first type of machine. Given that $r_1$ is the minimum reward among all machines and the linearity of the cost function with respect to $d_1$, we want to maximize $d_1$. As it can be seen from the optimization problem, we have $d_1 \leq \min \{\sum_{i=1}^{N} \epsilon_i c_i, \,  c_1^M \}$, and we can maximize $d_1$ by choosing this upper bound. Now that $d_1$ is fixed, we move on to $d_2$, where its upper bound is $d_2 \leq \min \{\sum_{i=1}^{N} \epsilon_i c_i - d_1, \,  c_2^M \}$, by repeating this process and finding the optimal $d_k$ in ascending order, we obtain the optimal strategy presented in Theorem \ref{thm:deployment}. 
\end{proof}

 This theorem implies that the cryptominer could potentially participate in ancillary service with the available capacity of all types of machines, but only use the less efficient devices to reduce its load when being deployed by certain ancillary service programs. Intuitively, this means that cryptominers are financially incentivized to keep their newer machines in use and help the grid achieve better load balancing by turning off their older ones. 
 
 %, which leads to a substantially reduced electronic waste and a more sustainable operation. 

\subsection{Optimal Ancillary Service Profile}
In the previous section, we determined the optimal deployment of different types of cryptomining devices given the total deployed capacity of ancillary service programs. In this section, the results of theorem \ref{thm:deployment} are used to find the optimal ancillary service profile. We present the following definition to capture the types of machines needed to cover the deployed capacity.

\begin{defn} 
For a given ancillary service profile $c \in \mathbb{R}^N $, and deployment ratio $\epsilon \in \mathbb{R}^N $, we define the \emph{"critical type"} denoted by $k_c$ as follows:
\begin{equation}
\label{kprime}
k_c= \arg\min\limits_{{ q \in \mathcal{K}  }} q, \quad s.t. \quad \sum_{i=1}^{N} \epsilon_i c_i \leq \sum_{k=1}^{q} c_k^M.
\end{equation} 
\end{defn}
Here $k_c$ captures the type of cryptomining devices with partially deployed capacity. All types of cryptomining devices less efficient than $k_c$ are fully deployed, and the more efficient ones are not deployed at all. Using the results of Theorem \ref{thm:deployment}, for any $k \in \mathcal{K}$ we have:
\begin{equation} 
		d_k \mbox{=} \begin{cases}
		c_k^M	\, & \mathrm{if } \quad  k < k_c,\\
  \sum_{i=1}^{N} \epsilon_i c_i - \sum_{j < k_c } c_j^M	\, & \mathrm{if } \quad  k = k_c,\\
			0 &  \mathrm{if } \quad  k>k_c ,
		\end{cases} 
\end{equation} 
By substituting this optimal deployment strategy in the cost function of the optimization problem \eqref{op:deploy}, we have: 
\begin{align} \label{cost}
cost(\epsilon, c)=&\sum_{k <k_c} r_k c_k^M + r_{k_c} \big( \sum_{i=1}^{N} \epsilon_i c_i-\sum_{k<k_c}c_k^M \big)- \sum_{i=1}^{N}c_ip_i \nonumber\\ 
   =& \sum_{k <k
_c} (r_k-r_{k_c}) c_k^M + \sum_{i=1}^{N} c_i( r_{k_c}\epsilon_i - p_i ).
 \end{align} 
This is a cost function that does not depend on $d_k$ anymore and is a function of the ancillary service profile $c$ and the deployment ratio $\epsilon$. Now, we finally obtain the profit maximization problem for cryptominers participating in ancillary services as follows:
\begin{subequations} \label{op:multi}
\begin{eqnarray}
	&{\min}& \EX_{_{\mathcal{\epsilon}}}[cost(\epsilon, c)]\\
	&\mbox{s.t.}\;& \sum_{i=1}^{N}c_i \leq C^M,  \label{feasible}\\ 
          &\mbox{vars.}\;& c_i \geq 0,  \quad  \quad \quad \quad \textit{for} \quad i \in \mathcal{N},
 \end{eqnarray} 
\end{subequations}
where the decision variables are $[c_i]_{i=1}^{N}$, and we are maximizing the expected profit with respect to $\epsilon$. Here, the constraint \eqref{feasible} ensures that the total committed capacity does not exceed the total combined capacity of all types of cryptomining devices.

\section{Algorithm Design}\label{sec:algorithm}

To solve the optimization problem in \eqref{op:multi}, first, we need to understand how the objective function depends on the decision variables. The cost function is the expected value of cost with respect to the random variable $\epsilon\in \mathbb{R}^N$. Here, we first consider the properties of the function $cost(\epsilon, c)$ for a fixed $\epsilon$, then we use those properties to solve the expected cost. Consider the feasible region for this optimization problem defined as follows:
\begin{equation}
  \mathcal{F}_c \coloneqq \{c \in \mathbb{R}^N | \quad c \succeq 0, \quad \sum_{i=1}^{N} c_i < C^M\}. \label{feasible-region}
\end{equation}

For any fixed $\epsilon$, there exist $K$ hyperplanes, each associated with a different $k_c \in [1:K]$ defined as follows:
\begin{equation}
\sum_{i=1}^{N} \epsilon_i c_i = \sum_{k=1}^{k_c} c_k^M, \quad for \quad k_c \in [1:K]
\end{equation}

These hyperplanes divide the feasible region $\mathcal{F}_c$ into at most $K$ regions. Depending on $c$, the value of $\sum_{i=1}^{N} \epsilon_i c_i$ falls within one of these regions, and all the points $c$ that fall in that region have the same critical type $k_c$. From an operational perspective, if $\sum_{i=1}^{N} \epsilon_i c_i$ falls within regions $k_c$, it means that all types of machines less efficient than $k_c$ are fully deployed and type $k_c$ machines are partially deployed. For a cryptominer with two types of cryptomining devices and two possible ancillary service programs, these regions are shown in Fig.~\ref{fig:regions}. In this figure, for a fixed $\epsilon$ the feasible region is divided into two parts, where in the orange one ($k_c=1$), only type 1 machines are needed to cover the deployed capacity. However, in the pink region ($k_c=2$), both type 1 and type 2 machines are deployed. It can be seen that for each region, the cost function is linear in terms of $c$ and for the whole feasible region is piecewise linear and convex. The following Lemma extends this observation to the general setting.
\begin{figure}
\centering
\includegraphics[width=0.6\columnwidth]{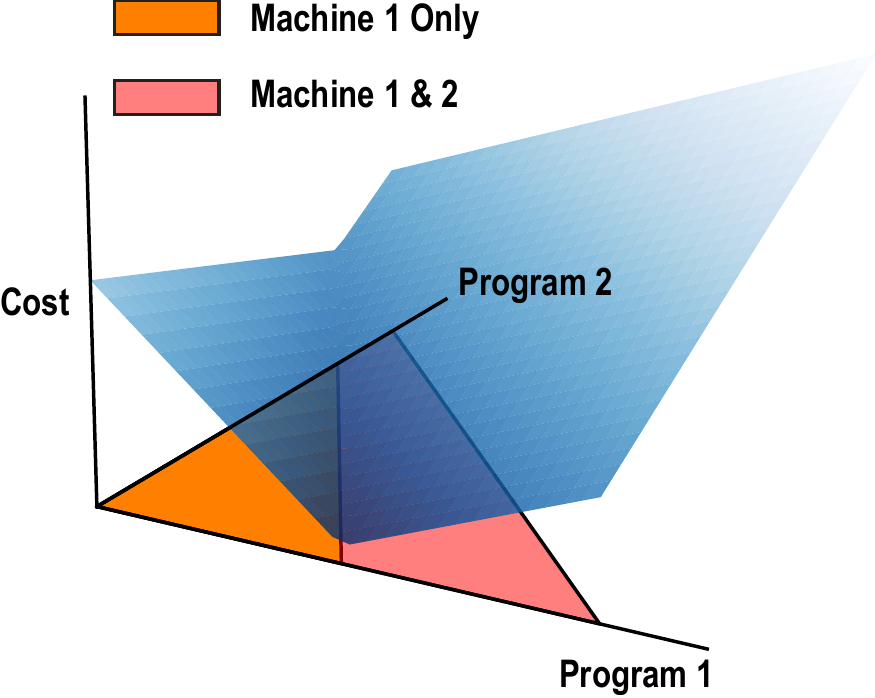}
\caption{\label{fig:regions}The cost function for participating in two ancillary service programs with capacity $c_1$ and $c_2$. For any given $\epsilon$, there are two regions. In region 1, only machine 1 is used, but in region 2, both types of machines are used to cover the deployment. It can be seen that the cost function is piecewise linear and convex.}
\end{figure}

\begin{lem}
\label{lem:convex}
For a fixed $\epsilon \in \mathbb{R}^N$, the cost function $cost(\epsilon,c)$ given in \eqref{cost} is piecewise affine and convex with respect to $c$.
\end{lem}
\begin{proof}

\textbf{Piecewise Affine:} Within each region, the critical type $k_c$ is fixed, hence the first part of the cost function $\sum_{k <k_c} (r_k-r_{k_c}) c_k^M$ is a constant, and the second part $\sum_{i=1}^{N} c_i( r_{k_c}\epsilon_i - p_i)$ is linear in terms of $c$. This makes the cost function affine within each region. 

\textbf{Convex:} The critical type $k_c$ depends on the value of $c$. Let us replace $k_c$ in the cost function with another variable $k' \in [1:K]$ that is independent of $c$. This creates a new function denoted as $cost(\epsilon, c, k')$, which depends on the value of $k'$ as follows: 

\begin{eqnarray}
    cost(\epsilon, c, k') = \sum_{k <k'} (r_k-r_{k'}) c_k^M + \sum_{i=1}^{N} c_i( r_{k'}\epsilon_i - p_i ) 
\end{eqnarray}

For any fixed $k' \in [1:K]$, this cost function is affine in terms of $c$. Hence there exist a total of $K$ affine functions, each defined by a different $k'$. Our goal is to show that for any fixed $\epsilon$ the cost function $cost(\epsilon, c)$ is the maximum of these $K$ affine functions, therefore it is a convex function \cite{cvxoptimization}:
\begin{eqnarray}
    cost(\epsilon, c) = \underset{k' \in [1:K]}{\max} cost(\epsilon, c, k')
\end{eqnarray}
\begin{comment}
\begin{eqnarray}
k_c = \underset{k' \in [1:K]}{\argmax} \big[\sum_{k \in \mathcal{K'}} (r_k-r_{k'}) c_k^M + \sum_{i=1}^{N} c_i( r_{k'}\epsilon_i - p_i )
\end{eqnarray}
\end{comment}

To find the maximum of the function $cost(\epsilon, c, k')$, let us consider its value changes when we increase $k'$ by one unit. Take two cost functions,  $cost(\epsilon, c, k')$ and $cost(\epsilon, c, k'+1)$, respectively. We compare the values of these two functions as follows:

\begin{eqnarray}
 && cost(\epsilon, c, k'+1) - cost(\epsilon, c, k') = \sum_{k <k'+1} (r_k-r_{k'+1}) c_k^M \notag\\
 && - \sum_{ k <k'} (r_k-r_{k'}) c_k^M  + \sum_{i=1}^{N} c_i\epsilon_i (r_{k'+1}-r_{k'})\notag \\
 && = \sum_{k <k'+1} (r_{k'}-r_{k'+1}) c_k^M+\sum_{i=1}^{N} c_i\epsilon_i (r_{k'+1}-r_{k'})\notag \\
 && = (r_{k'+1}-r_{k'}) [\sum_{i=1}^{N} c_i\epsilon_i -\sum_{k <k'+1}c_k^M ] 
 \end{eqnarray} 
The term $(r_{k'+1}-r_{k'}) $ is always positive, hence if $\sum_{i=1}^{N} c_i\epsilon_i > \sum_{k <k'+1}c_k^M $, the value of the function increases by going from $k'$ to $k'+1$. On the other hand, if $\sum_{i=1}^{N} c_i\epsilon_i \leq \sum_{k <k'+1}c_k^M $, the value decreases by increasing $k'$.Therefore to find the $k'$ that maximizes the function $cost(\epsilon, c, k'+1)$, we need to keep increasing $k'$ until $\sum_{i=1}^{N} c_i\epsilon_i \leq \sum_{k <k'+1}c_k^M $. This is the same as choosing $k'=k_c$, and $k'=k_c$ is the maximizing value over all possible values of $k'$. Hence, the cost function defined using $k_c$ is the pointwise maximum of $K$ affine functions, and it is convex. This completes the proof.
\end{proof}
Now that we understand the properties of the cost function for a fixed $\epsilon$, we present the following theorem regarding the properties of the expected cost. 

\begin{thm} The expected cost function $\EX_{_{\mathcal{\epsilon}}}[cost(\epsilon, c)]$ is a convex function in terms of the decision variable $c$, and the solution found by the stochastic subgradient descent algorithm converges to the optimal solution in expectation. 
\label{thm:convexity}  
 \end{thm}

\begin{algorithm}
\caption{Stochastic Subgradient Descent}\label{two type algorithm}
\begin{algorithmic}[1]
\Procedure{Find the optimal ancillary service profile} {}
\State Given: $[p_i]_{i=1}^N,[r_k]_{k=1}^K, [c_k^M]_{k=1}^K$,$J$,$[\alpha_j]_{j=1}^J$
\State Initialize $c \in \mathcal{F}_c$,$j=0$
%\State Initialize $[\alpha_j]_{j=1}^J$ such that equation \eqref{alpha} holds
%\State Initialize , 
%\While{$j\leq J$}
%\State Compute total $cost$ for all program
%\ForEach{$(\epsilon_1,\ldots,\epsilon_N)$ sample}
\For{$j$ in $[1:J]$}
\State Random sample $(\epsilon_1,\ldots,\epsilon_N)$ a total $M$ times.
\State $ c^{(j+1)} \gets c^{(j)}-\alpha_j*\left[\cfrac{1}{M}\sum_{m=1}^M r^{(m)}\epsilon^{(m)} - p\right]$
%\frac{\partial cost(\epsilon,c)}{\partial c_i}$ 
%\Comment{Gradient Descent}
\State Project the $c^{(j+1)}$ onto the feasible set.  
\EndFor
%\EndFor

%\EndWhile
\State \textbf{return $c=\frac{1}{J}\sum{j=1}^Jc^{(j)})$}
\EndProcedure

%\Procedure{Determine the optimal deployment for each type of mining device}{}
%\State Order $k$ devices where $k \in \mathcal{K}$ from low efficient to high efficient.
%\State \eqref{deployed-calc}
%\EndProcedure
\end{algorithmic}
\end{algorithm}

\begin{figure} 
\centering
\includegraphics[width=0.8\columnwidth, trim={7em 4em 7em 0em}]{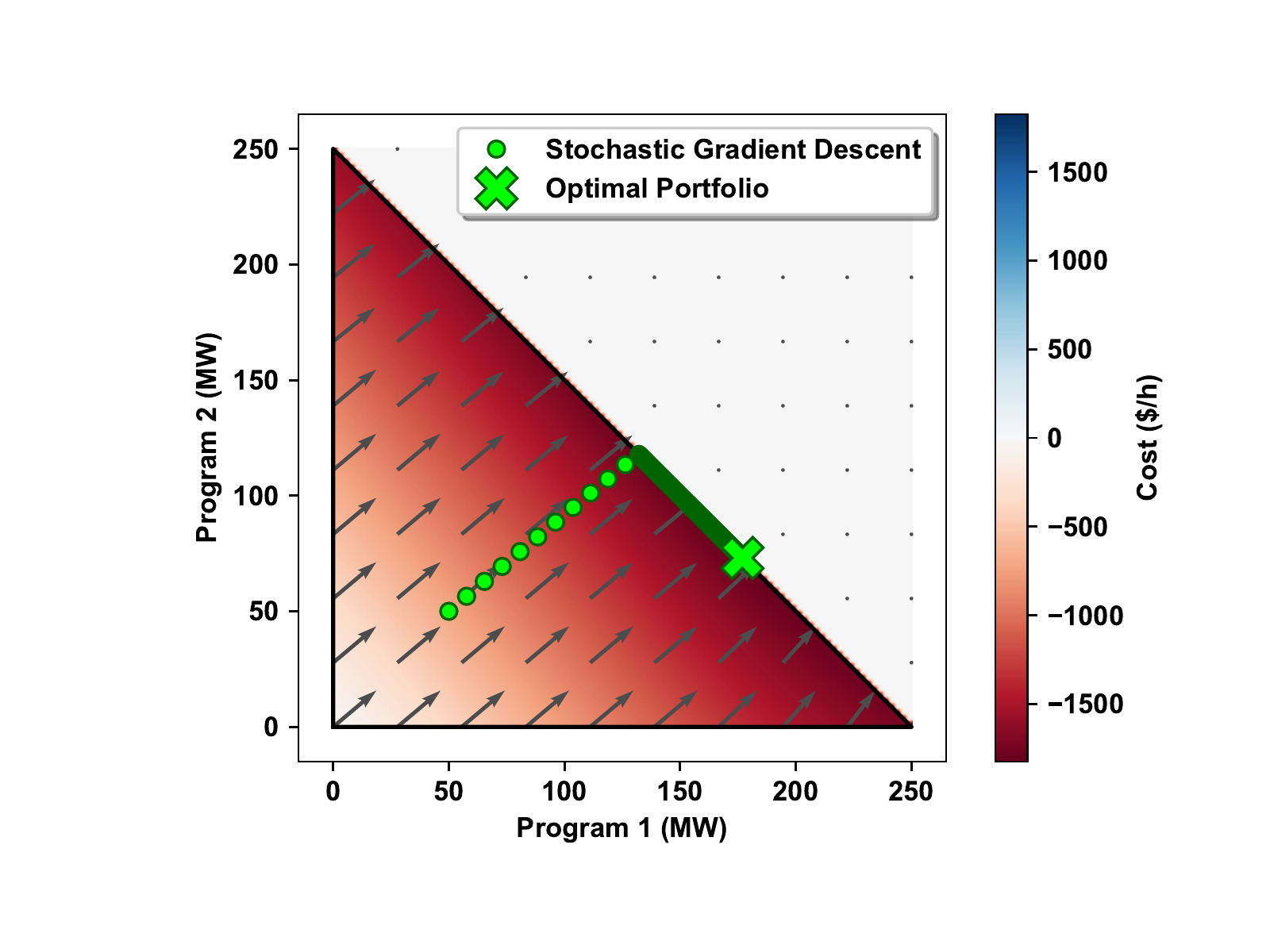}
\caption{Performing stochastic subgradient descent to obtain the optimal ancillary service profile.}
\label{fig:SGD}
\end{figure}

\begin{proof} In Lemma \ref{lem:convex}, we have shown that for any fixed $\epsilon$ in the probability distribution of $\epsilon$, the cost function $cost(\epsilon,c)$ is convex in terms of $c$. Hence, according to \cite{cvxoptimization} (chapter 3.2.1) the expected cost function $\EX_{_{\mathcal{\epsilon}}}[cost(\epsilon, c)]$ is also convex in terms of $c$. We note that in the optimization problem \eqref{op:multi}, both the objective function and the feasible region are convex. So we can use stochastic subgradient descent to obtain the optimal solution. The algorithm is presented in Algorithm \ref{two type algorithm}. 
%To perform stochastic subgradient descent, at each iteration, we use

%\textbf{Stochastic Gradient Descent:}  
To perform stochastic subgradient descent on this convex function, we randomly sample $\epsilon$ a total of $M$ time in each iteration, according to the known $\epsilon$ distribution. 
Denote the $m^\text{th}$ sample as $\epsilon^{(m)}\in\mathbb{R}^N$. For each sample $\epsilon^{(m)}$, we can compute the subgradient for the function $cost(\epsilon^{(m)},c)$, which is given by:
\begin{equation}
     \nabla_c \text{cost}(\epsilon^{(m)}, c) = r_{k_c}\epsilon^{(m)} - p.
     \label{eqn:maxrp}
\end{equation}
Clearly, the average of these sample subgradients would form an unbiased estimate for the subgradient of the objective function, i.e., 
\begin{align}
\nabla_c \mathbb{E}_\epsilon \left[\text{cost}(\epsilon, c)\right] = \mathbb{E}_\epsilon \left[\cfrac{1}{M}\sum_{m=1}^M \nabla_c \text{cost}(\epsilon^{(m)}, c)\right].
%    \nabla_c \mathbb{E}_\epsilon \left[\text{cost}(\epsilon, c)\right] = \mathbb{E}_\epsilon \left[\nabla_c \text{cost}(\epsilon, c)\right] \approx \cfrac{1}{M}\sum_{m=1}^M \nabla_c \text{cost}(\epsilon, c) \Bigr|_{\epsilon \gets \epsilon^{(m)}}
\end{align}
A crude bound on the expected squared 2-norm of the gradient estimate is given as: 
\begin{align}
    \mathbb{E}_\epsilon \big{\|}\cfrac{1}{M}\sum_{m=1}^M \nabla_c \text{cost}(\epsilon^{(m)}, c)\big{\|}^2_2\leq N\max(r^K,p_{\max})^2,
\end{align}
where $p_{\max}$ is the maximum value of all $p_i$, due to the relation given in (\ref{eqn:maxrp}).

%Let us consider our cost function $\EX_{_{\mathcal{\epsilon}}}[cost(\epsilon, c)]$ as $f(c)$, which is not a function of $\epsilon$ anymore. To perform stochastic gradient descent, 

Let the step sizes $\alpha_j$'s be chosen such that:
\begin{equation} \label{alpha}
\alpha_j =\frac{D}{\sqrt{N}\max(r^K,p_{\max})\sqrt{j}}, \quad j=1,2,\ldots,J.
%    \alpha_j \geq 0, \quad \sum_{j=1}^{\infty} \alpha_j^2 = ||\alpha||_2^2 \leq \infty \quad\sum_{j=1}^{\infty} \alpha_j = \infty.
\end{equation}
Using standard result in stochastic gradient descent \cite{duchi2018introductory}, we have:
\begin{eqnarray}
\EX \left[f(\bar{c})\right]-f(c^*)  \leq \frac{3D\sqrt{N}\max(r^K,p_{\max})}{2\sqrt{J}}, \label{eq:bound}
\end{eqnarray}
Where $\bar{c}=\frac{1}{J}\sum_{j=1}^Jc^{(j)}$, and the radius $D$, which is an upper bound on the diameter of the feasible region, is given as:
\begin{equation}
D =\begin{cases}
        C^M&
        \text{if} \quad N=1 \\
         \sqrt{2}C^M &
        \text{if}\quad N \geq 2.
\end{cases}   
\end{equation}
The right-hand side of equation \eqref{eq:bound} converges to zero as $J$ goes to infinity. Hence, the expected function value of the average of the points found so far converges to the optimal value $f(c^*)$, which completes the proof.
\end{proof}

In Fig.~\ref{fig:SGD}, stochastic subgradient descent is performed for a two dimensional case ($N=2$). As it can be seen, due to the convexity of the function, by starting anywhere in the feasible region and moving downhill using the gradient, we obtain the optimal ancillary service profile $c^* \in \mathcal{F}_c$.

\section{Dependent Programs: Frequency Regulation}\label{sec:dependent programs}
In this section, we investigate an important family of ancillary services called frequency regulation. Grid frequency is constantly fluctuating due to unbalanced load and generation, and frequency regulation is designed to maintain system frequency at a desired level. There are two types of regulation: regulation up (reg-up) and regulation down (reg-down). While participating in reg-up requires load owners to reduce their demand during certain hours and bring up the frequency. In reg-down, they need to increase their demand and use the over-generated power to bring down the frequency. In this sense, reg-down is different from other ancillary service programs. 
\begin{itemize}
    \item \textbf{Reg-up}: As previously discussed, if the cryptominer is not participating in any ancillary service program, it keeps mining with maximum capacity $C^M$. Consider a cryptomining facility that participates in the reg-up program with capacity $c_{up}$ and gets deployed by the system operator to reduce its consumption by $\epsilon_{up}c_{up}$. It is required to adjust its normal operating point from $C^M$ to $C^M-\epsilon_{up}c_{up}$. Hence, the total reduced power consumption of the cryptomining devices is $\sum_{k=1}^{K} d_k = \epsilon_{up}c_{up}$, which is covered using different machines. 

\item \textbf{Reg-down}: Consider a cryptomining facility participating in the reg-down program with capacity $c_{dn}$. This facility should have the ability to increase its power consumption by $c_{dn}$. Hence, it needs to adjust its normal operational point from $C^M$ to a new operating point $C^M-c_{dn}$, which gives it a headroom of size $c_{dn}$. If the cryptominer is deployed to increase its power consumption by $\epsilon_{dn}c_{dn}$, its normal operating point changes to: 
\begin{equation}
    C^M-c_{dn} + \epsilon_{dn}c_{dn} = C^M-(1-\epsilon_{dn})c_{dn} 
\end{equation}
In other words, the total reduced power consumption of all the cryptomining devices is $\sum_{k=1}^{K} d_k = (1-\epsilon_{dn})c_{dn}$.
\end{itemize}

As can be seen, reg-up is different from other ancillary services. It should be noted that we can use an almost identical program as in \eqref{op:multi} to find the optimal participation in the reg-down programs, with the exception that if program $i$ is a reg-down program, the $\epsilon_i$ term in the optimization problem needs to be replaced with $1-\epsilon_i$. 

In the following, we use the fact that reg-up and reg-down programs are not independent to define their joint probability and obtain a closed-form cost function for this special case. Let us consider a scenario where the cryptomining facility has two types of machines, with one being more efficient $r_1 < r_2$. This cryptomining facility can only participate in reg-up and reg-down programs, and we denote $c_1$, and $c_2$ as the committed capacity in the reg-up and down programs, respectively. We also denote $\epsilon_1$, and $\epsilon_2$ as the deployment ratio of the reg-up and down programs, respectively. In solving the optimization problem in the general setting, we use the probability distribution of $\epsilon$, which is a joint distribution over all programs. In the special case of regulation services, reg-up and reg-down are dependent on each other. In particular, at any moment, only one of them is being deployed to either increase or decrease the frequency. Let us assume that the reg-down program is deployed with probability $\theta$, and the reg-up program is deployed with probability $1-\theta$. If both programs are deployed, we assume the deployment ratio follows a truncated exponential distribution. More precisely, we consider the following joint probability distribution for $0\leq \epsilon_1 \leq 1$ and $0\leq \epsilon_2 \leq 1$: 
\begin{equation} \label{eq:probability}
		PDF(\epsilon_1, \epsilon_2) = 
		\theta *\delta(\epsilon_1)  \underbrace{\frac{\lambda_2 e^{-\lambda_2 \epsilon_2}}{1-e^{-\lambda_2}}}_{\textit{reg-down}} + (1-\theta) * \delta(\epsilon_2) \underbrace{ \frac{\lambda_1 e^{-\lambda_1 \epsilon_1}}{1-e^{-\lambda_1}} }_{\textit{reg-up}},
\end{equation} 
where terms indicated by reg-up and reg-down show the truncated exponential distribution for each program. Our goal is to solve the following optimization problem:

\begin{subequations} \label{op:regup-down}
\begin{eqnarray}
	&\underset{c}{{\min}}& \EX_{_{\mathcal{\epsilon}}}[cost_{reg}(\epsilon, c)]\\
	&\mbox{s.t.}\;& c \in \mathcal{F}_c,
 \end{eqnarray} 
\end{subequations}
Where $cost_{reg}(\epsilon, c)$ is a special case of the general cost function because only reg-up and reg-down programs are considered, with the reg-down negation as discussed earlier. Using the above probability distribution, we calculate the closed form for the expected cost function of these programs as follows:

\begin{thm} Consider a cryptomining facility that can only participate in reg-up and reg-down programs. Given the joint probability distribution \eqref{eq:probability} for the deployment rate of programs, we have:

\begin{eqnarray}
    \EX_{_{\mathcal{\epsilon}}}[cost_{reg}(\epsilon, c)] = 
    \begin{cases}
        Q_1&
        \text{if} \quad c_1+c_2 < c_1^M \\
        Q_2 &
        \text{if} \quad c_2 \geq c^M_1 \\
        Q_3&
        \text{if} \quad otherwise,
    \end{cases}
\end{eqnarray}
where 
\begin{eqnarray}
&& Q_{1} = c_1[(1-\theta)r_1\EX[\epsilon_1]-p_1] +c_2[r_1(1-\theta\EX[\epsilon_2])-p_2] \notag\\
&& Q_2 = Q_{1}+(r_1-r_2)\bigg( \theta \big[\frac{\lambda_2c_1^M + c_2 (1-\lambda_2-e^{-\lambda_2 (1- c_1^M/c_2)}}{\lambda_2(1-e^{-\lambda_2})}) \big] \notag \\
&& + (1-\theta) \big[ c_1^M-c_2 -c_1 \EX[\epsilon_1]  \big] \bigg) \notag \\
&& Q_3 = Q_{1}+ (r_1-r_2)(1-\theta) \notag \\
&&  \big[\frac{\lambda_1(c_2-c_1^M)+ c_1 (1+\lambda_1-e^{-\lambda_1 (c_1^M-c_1-c_2)/c_1})}{\lambda_1(e^{\lambda_1}-1)}) \big].
\end{eqnarray}

\label{thm:reg}  
\end{thm}

\begin{proof}
    Refer to Appendix \ref{appendixA}.
\end{proof}
Given the closed form and the convexity of this cost function, we can easily solve the optimization problem \eqref{op:regup-down} using standard convex optimization solvers.
\section{Single Machine Type facilities}
\label{sec:risk_aware_formulation}
\subsection{Risk-Oblivious Participation}
In this section, we consider a special setting where the cryptomining facility only uses one type of crptomining device with reward $r$. In this case, the optimization problem \eqref{op:multi} reduces to the following:

\begin{subequations} \label{eq:OPone} 
\begin{eqnarray}
	&\underset{c}{\min}& \EX_{_{\mathcal{\epsilon}}}[\sum_{i=1}^{N} c_i \epsilon_i r-\sum_{i=1}^{N} c_i p_i ]\\
	&\mbox{s.t.}\;& c \in \mathcal{F}_c,
 \end{eqnarray} 
\end{subequations}
Where $c_i \epsilon_i r$ is the loss of cryptomining revenue due to deployment in ancillary service program $i$. As can be seen, the deployment ratio $\epsilon$ no longer appears in the constraints, which makes it easier to find the optimal solution. The cost function in this optimization problem reduces to taking the expected value with respect to each program as follows:
\begin{equation}
\sum_{i=1}^{N} c_i \big(r \EX[\epsilon_i ]-p_i  \big).     
\end{equation}
This optimization problem is linear with respect to $c_i$, and $\EX[\epsilon_i]$ is obtained using the distribution of $\epsilon_i$ based on historical data. The solution to this optimization problem is presented in the following Theorem:

\begin{thm} The optimal ancillary service capacity of the crypotominer for participation in each program $i \in \mathcal{N}$ is: 
\begin{equation} 
		c_i \mbox{=} \begin{cases}
		C^M\mathcal{U}\big(p_{i^*} -  r\EX[\epsilon_{i^*} ] \big) \, & \mathrm{if } \quad  i = i^*,\\
  0	\, & \mathrm{if } \quad  i \neq i^*,
		\end{cases} 
\end{equation}
where 
\begin{equation}
i^*= \arg\min\limits_{{ i \in \mathcal{N}  }} \big[r\EX[\epsilon_i] -p_i \big],
\end{equation}
represents the ancillary service program with the minimum per-unit cost, and $\mathcal{U}(x)$ is a unit step function with value one for $x \geq 0$, and zero otherwise. 
\label{thm:linear}  
 \end{thm} 

\begin{proof}
We use the principles
of the simplex method, where the inequality constraints define a polygonal region, and the solution is at one of the vertices. Hence, there is at most one non-zero $c_i$, and to find the optimal solution, one should determine the ancillary service program $i^*$ with the minimum cost $ r \EX[\epsilon_i ]- p_i$. If this cost is negative (profit is positive), the cryptominer participates in program $i^*$ with maximum capacity $c_{i^*}=C^M$. Otherwise, if the profit is negative, the revenue of the ancillary service is not worth the loss of revenue due to not mining cryptocurrency, so the cryptominer does not participate in any program $c_{i^*}=0$.  
\end{proof}
In the following section, we consider the risk associated with each program and rewrite the optimization problem for a cryptominer with one type of machine to jointly minimize the expected cost and the risk of participation in ancillary services.

\subsection{Risk-Aware Participation}
In the modern power system, due to the volatile and uncertain nature of both load and generation, the grid experiences significant supply-demand mismatches and price variations. Therefore, participation in ancillary services always comes with potentially large risks, which might not be acceptable to the risk-averse nature of some cryptocurrency mining facilities with their steady source of income. In this section, we present a risk-aware optimization problem where the cryptominer can maximize its profit while controlling the risk below its desired level. At each time slot, we solve the following optimization problem:

 \begin{subequations}
\begin{eqnarray}\label{op:risk-aware} 
	&{\min}& \EX_\epsilon \big[\sum_{i=1}^{N} c_i\big( \epsilon_i r-p_i \big) \big]+ \lambda  \textit{Var} \, \big[\sum_{i=1}^{N} c_i\big(\epsilon_i r-p_i \big) \big] \\
 &\mbox{s.t.}\;& c \in \mathcal{F}_c,
\end{eqnarray}
 \end{subequations}
where the second term captures the uncertainty of the information, and its impact is limited by the control variable $\lambda$. Considering $Var(X)= \EX(X^2)- \EX^2(X)$, this problem formulation can be further reduced to the following:
 \begin{subequations}

\begin{eqnarray}\label{op:risk-aware2}
	&{\min}&\sum_{i=1}^{N} \big[ c_i(\EX [\epsilon_i] r-p_i ) + \lambda  c_i^2 r^2 \textit{Var}\, [\epsilon_i] \big] \\
	&\mbox{s.t.}\;& c \in \mathcal{F}_c,
	\end{eqnarray} 
  \end{subequations}
which is a standard quadratic optimization with linear constraints and can be solved using standard solvers. Here, $\lambda$ is a weight that determines the importance of risk minimization compared to profit maximization, and a larger $\lambda$ leads to a more conservative strategy with less participation in ancillary services. For example, if we have two different ancillary service programs with the same expected return but one has a higher variance, the solution to this optimization problem chooses the program with less variations in profit. 

\section{Online Optimization}\label{sec:online}
In this section, we formulate cryptominers' participation in ancillary services as an online optimization problem. We consider a multi-round game between a learner and an adversary. In our setup, the cryptominer plays the role of a learner trying to minimize the cost of participating in ancillary services, and the grid operator is the adversary. At each round $t$, the cryptominer tries to obtain the ancillary service profile $c(t)$ that solves:
\begin{subequations} 
	\begin{eqnarray}
	&\underset{c(t)}{\min}& cost_t(c(t)) \\
	&\mbox{s.t.}\;& c(t) \in \mathcal{F}_c,
 \end{eqnarray} 
\end{subequations}
where $\mathcal{F}_c$ is the feasible region for the ancillary service profile defined in \eqref{feasible-region}, and:
\begin{align} 
&cost_t(c(t)) = \notag\\
&\,\,\sum_{k <k
_c(t)} (r_k(t)-r_{k_c}(t)) c_k^M + \sum_{i=1}^{N} c_i(t)( r_{k_c}(t)\epsilon_i(t) - p_i(t) ). \label{costonline}
 \end{align} 
These cost functions are functions of $c(t)$, and they are announced sequentially. At each round $t$, the cryptominer fixes its hourly ancillary service profile $c(t)$. And depending on the deployment rate $\epsilon(t)$, cryptomining reward $r_k(t)$, and ancillary service price $p_i(t)$ chosen by the adversary, cryptominer suffers the $cost_t(c(t))$. To solve this online optimization problem, we use a regret minimization framework, where the notion of regret is used to measure the performance of our online algorithm. In this framework, the regret of the cryptomining facility after $T$ time slots is defined as: 
\begin{eqnarray} 
R_T= \sum_{t=1}^{T} cost_{t}(c(t)) - \underset{{ c \in \mathcal{F}_c}}{\min} \sum_{t=1}^{T} cost_{t}(c).
\end{eqnarray}
The regret $R_T$ quantifies the difference in performance between the ancillary service profile generated by an online algorithm for the cryptominer, and the performance of the suboptimal, round-to-round constant, ancillary service profile $c^*$, where:
\begin{equation}
c^*= \arg\min\limits_{{ c \in \mathcal{F}_c  }} \sum_{t=1}^{T} cost_{t}(c) .
\end{equation}
We use the online gradient descent (OGD) algorithm \cite{zinkevich2003online} to minimize regret by updating the ancillary service profile $c(t)$. At each round $t$, the regret minimization algorithm generates the update without knowing the current objective function and only using the information of the previous round. For $t=1$ to $T$, the OGD algorithm updates the ancillary service profile iteratively as follows:
\begin{eqnarray}
    && c'(t+1) = c(t) -\eta_t \nabla cost_t(c(t)) \\
    && c(t+1) = \textbf{proj}_{\mathcal{F}_c} (c'(t+1))
\end{eqnarray}
In each iteration, the algorithm uses the gradient of the cost function from the previous step and updates the ancillary service profile $c(t)$. This may result in finding a point $ c'(t+1)$ that is outside the feasible region. In such cases, $c'(t+1)$ is projected into the convex feasible region $\mathcal{F}_c$, where the projection is defined as:
\begin{equation}
    \textbf{proj}_{\mathcal{F}_c} (c'(t+1))= \arg\min\limits_{{c \in \mathcal{F}_c }} ||c'(t+1) - c||.
\end{equation}
It should be noted that in the case of ancillary services, each hour has a very different price pattern and deployment rate. Hence, instead of learning a single cost function for all the hours of the day, we learn a separate cost function for each hour. In other words, when calculating $c(t)$ for a certain hour of the day, we use $c(t-1)$ and $cost_t(c(t-1))$ of the same hour during the last day. 

While the cost function changes every round and can be very different from the costs observed so far, the algorithm induces a sublinear regret. The following Theorem summarizes the convergence of the ancillary service profile updates generated by the algorithm.

\begin{thm}\label{online}
Let us assume that $r_k(t)\leq r_{max}$, and $p_i(t)<p_{max}$ for all $t \in \mathcal{T}$,  $k\in \mathcal{K}$, and $i \in \mathcal{N}$. The online subgradient descent with step sizes ${\eta_t = \frac{D}{G\sqrt{t}}, \, \, t \in [1:T] }$ guarantees the following:
\begin{equation}
    R_T \leq \frac{3}{2} GD \sqrt{T} = 3C^M \sqrt{\frac{TN}{2}} \max \{r_{max}, p_{max}\},  
\end{equation}
where 
\begin{equation}
G=\sqrt{N} \max \{r_{max}, p_{max}\}, \, \, \, and \, \, \,
 D =\begin{cases}
        C^M&
        \text{if} \quad N=1 \\
         \sqrt{2}C^M &
        \text{if}\quad N \geq 2.
\end{cases}   
\end{equation}
\end{thm}

\begin{proof} Refer to Appendix \ref{proof:online}.
\end{proof}

This result shows that the average regret, i.e., $R_T/T$, converges to zero as $T \to \infty$. In other words, the data center is not well-informed at the beginning, and its decisions are not close to optimal. As the number of days increases, the average performance of the ancillary service profile generated by the algorithm converges to that of the optimal round-to-round invariant profile ($c^*$). 
\section{Numerical Experiments}\label{sec:experiment}
In this section, we present a Texas-based case study of the algorithms introduced  earlier, wherein we utilize real-world cryptocurrency and electricity market historical records to validate our proposed approaches.

The rest of this section is organized as follows. First, we introduce the simulation setup, including the source of the historical records and assumptions made. Second, we present simulation results for the offline optimization cases, including the special cases of risk-aware profile optimization and frequency regulation (reg-up/down) capacity co-optimization and the more general stochastic subgradient descent-based optimization with arbitrary demand response programs. Finally, we present results for online optimization, wherein we demonstrate how our algorithm learns to minimize average regret over time.

\subsection{Simulation Set-up}
We rely on a subset of publicly available ERCOT historical records to sample electricity prices \cite{ERCOTRTM2022} along with ancillary service prices and deployment rates \cite{ERCOTAS2022}. Furthermore, we use Bitcoin price data \cite{YAHOOBIT2022} and a record of mining difficulty over the last few years to estimate potential mining-based costs and rewards. 

We estimate the amount of mining energy required to earn one Bitcoin as follows. For state-of-the-art machines, energy efficiency ranges from roughly 0.02 to 0.03 J/GH. At the current mining reward of 6.25 Bitcoin per mined block (halved every four years, next in Apr 2024), and a mining difficulty level ranging from 15 to 30 trillion (unitless) from 2020 to 2022, we estimate a range of 50 to 150 MWh of energy required to mine one Bitcoin. Based on that, we consider two sets of machines used by a given facility for the remainder of our simulations, one set averaging at 110 MWh/Bitcoin with a total capacity of 100 MW and another (less efficient one) at 130 MWh/Bitcoin with a total capacity of 150 MW. Consequently, for any given hour, a facility in our numerical experiments is assumed to be able to participate with no more than 250 MW in capacity.

For each of the following simulation cases, we rely on one or two machine types and profile optimization over two programs. In the special case of Sections \ref{sec:risk_aware_results} and \ref{sec:reg_up_dn_results}, our results are based on analytical solutions given earlier, whereas in Sections \ref{sec:general_sgd_results} and \ref{sec:online_results}, stochastic subgradient descent is used to determine the optimal profile for a given scenario.

Finally, for multiple experiments described in this section, we model \textit{price-responsive} programs as follows. For some threshold, say \$60/MWh, if and only if real-time wholesale electricity prices exceed this threshold, the grid operator deploys 100\% of the committed capacity for this program. That is, $\epsilon\in\{0,1\}$ for this price-responsive program.

\subsection{Offline Optimization Cases}

\subsubsectionCUSTOM{Risk-Aware Participation}
\label{sec:risk_aware_results}
In this section, we consider the special case of only one mining machine type, as presented in Section \ref{sec:risk_aware_formulation}. Specifically, we examine the risk-aware case, wherein a control variable $\lambda \geq 0$ is employed to mitigate the risk taken during profile optimization. Since the problem posed in Eq. (\ref{op:risk-aware2}) can be expressed as a standard quadratic convex optimization problem, we can obtain an analytical solution to it. From the results shown in Fig. \ref{fig:risk}, based on an experiment conducted over one week in the summer of 2022 with two demand response programs, price-responsive and reg-up, we make the following observations. When $\lambda$ is increased above zero, i.e. less risk is taken, we observe less fluctuations in the profit, yet less average profit as well, as shown in the legend of the figure ($\lambda=0.0004$). This means that while hours of negative profit were avoided (e.g. night time as shown in the figure), many opportunities to earn profit were also squandered due to the particular risk-averse approach adopted.

\begin{figure}[htbp!]
\centering
\includegraphics[width=\columnwidth]{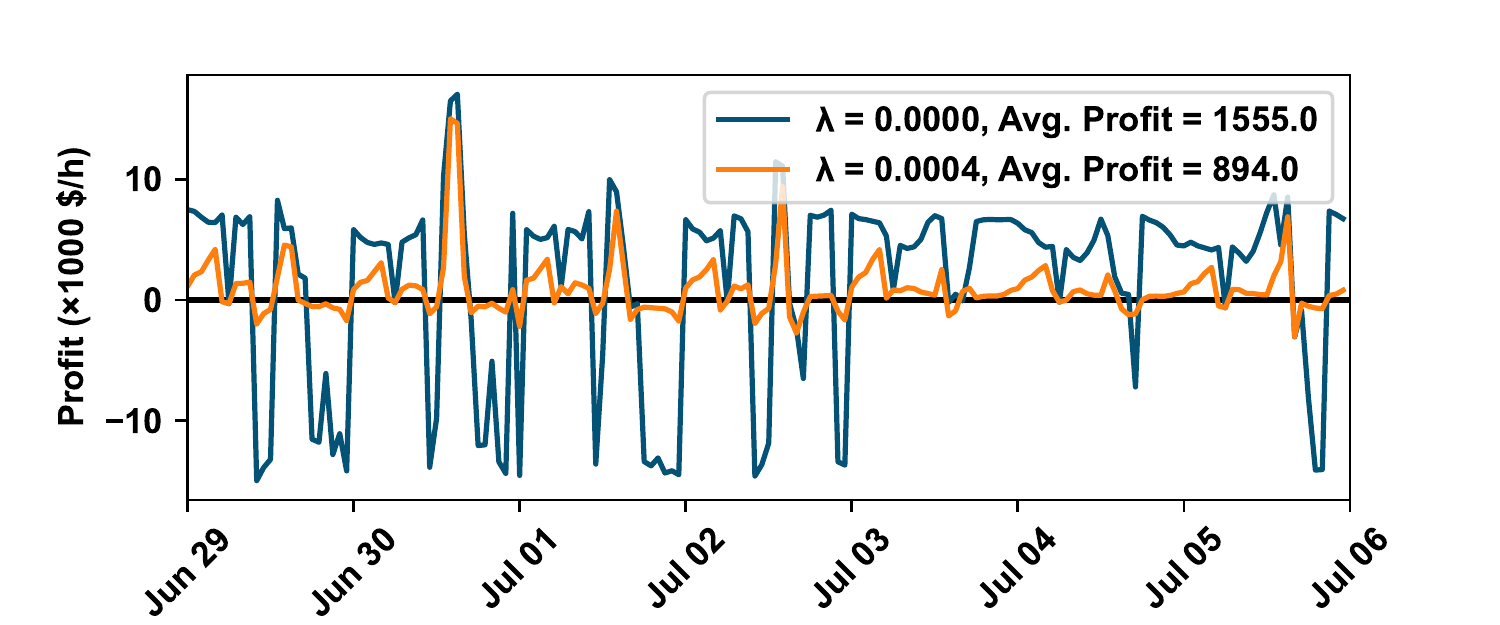}
\includegraphics[width=\columnwidth]{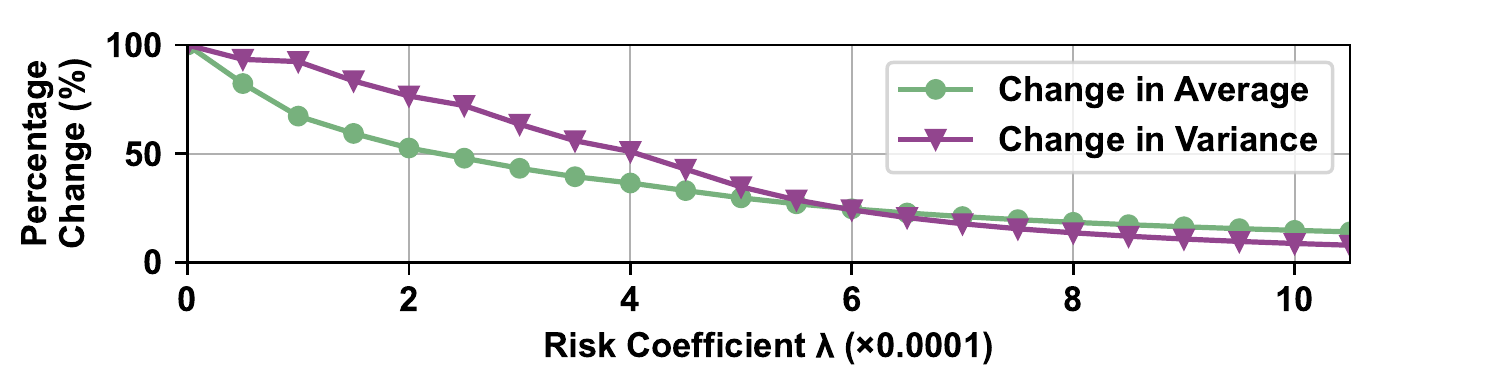}
\caption{(top) Profits using summer 2022 data under two approaches: risk-unaware ($\lambda=0$) and risk-aware ($\lambda=0.0004$). (bottom) The trade-off between average profit and variance for increased $\lambda$.}
\label{fig:risk}
\end{figure}

% \begin{figure*}
\begin{figure*}[htbp!]
% \begin{figure*}[!tbh]
    \centering   
    \includegraphics[width=0.95\textwidth, trim={12em 2em 12em 2em}]{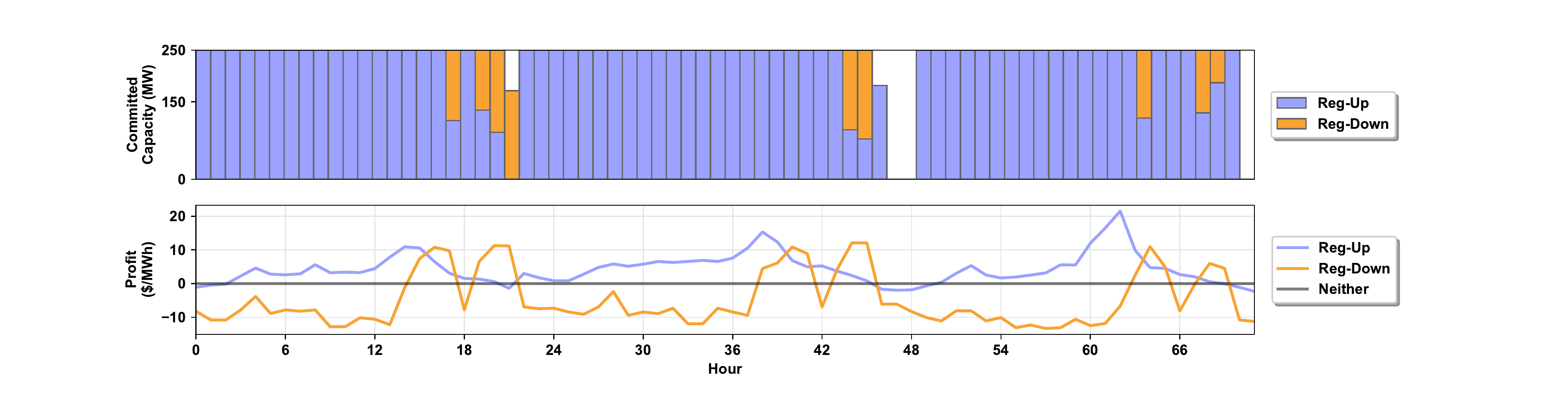}
    \caption{(top) Selected profile and (bottom) marginal profit per reg-up/down program for each hour on May 5-7, 2020.}
    \label{fig:reg_up_dn_choice_and_reason}
\end{figure*}

\subsubsectionCUSTOM{Dependant Programs - Frequency Regulation}
\label{sec:reg_up_dn_results}
In this section, we consider the special case of a mixture of reg-up and reg-down ancillary service programs, introduced in Section \ref{sec:dependent programs}. This case is special because: (a) deployment rates $\epsilon_{up}$ and $\epsilon_{dn}$ are not independently distributed - at most one can be positive and the other is forced to zero. (b) we have an analytical solution for optimal profile selection, as explained in Section \ref{sec:dependent programs}, based on the assumption that the underlying probability density functions for each program are exponential. Historical records of reg-up/down deployment at ERCOT reveal that an exponential distribution approximates the distribution for $\epsilon_{up}$ and $\epsilon_{dn}$ very closely, with expected values of 18\% and 27\%,  respectively,  over the duration of the experiment.

The results in Fig. \ref{fig:reg_up_dn_choice_and_reason} reveal that for many hours, reg-up is most profitable. However, the optimal strategy is not the same for all hours. For example, during hours 17, 19 and 20, the optimal profile is a mixture of reg-up/down, and during hours 47 and 48, choosing neither is better. In conclusion, the results reveal that optimal profile selection is not the same for each hour, and simply adopting a profile with a fixed split between the different programs is likely sub-optimal. We further explore this discovery in the next subsection for a more general case where the analytical solution is unknown.

\begin{figure}[htbp!]
\centering
\includegraphics[width=.9\columnwidth, trim={4em 2em 1em 2em}]{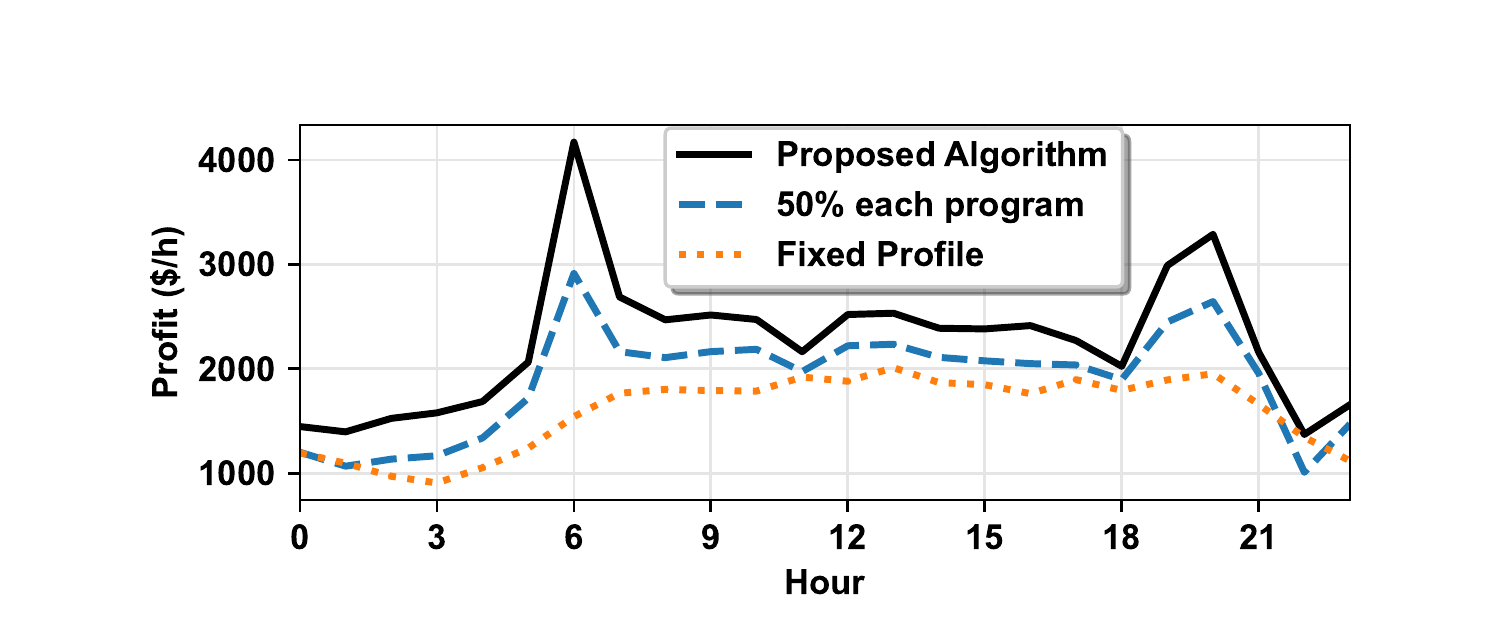}
\caption{Hourly profits averaged over one week in April 2020 using our algorithm versus alternative approaches.}
\label{fig:prices_under_alternatives}
\end{figure}

\subsubsectionCUSTOM{General Setting Using Stochastic Subgradient Descent}
\label{sec:general_sgd_results}
In this section, we examine a broader subset of profile selection problems for which the exact analytical solution is unavailable. Despite such constraint, we are still guaranteed that the optimization problem is convex. We utilize the stochastic subgradient descent method, as discussed in Section \ref{sec:problem}, to find the optimal solution within the feasible region.

In this specific case study, we consider the same pair of programs used in Section \ref{sec:risk_aware_results}, namely: price-responsive and reg-up, with the experiment conducted over one week in April of 2020 and results shown in Fig. \ref{fig:prices_under_alternatives}. We compare our proposed algorithm to the following alternative approaches: (a) an hour-independent profile, labeled \textit{Fixed Profile} in the figure, where we use stochastic subgradient descent to determine an optimal profile, but the same for all hours. (b) a profile which splits participation evenly across programs, hence the label of \textit{50\% each program}. 
%Surprisingly, even though the latter approach requires no search, it performs somewhat close to the proposed approach and matches its temporal patterns, albeit sub-optimal. 
%In contrast, 

The 50-50 strategy outperforms the hour-independent strategy for almost all the hours, further supporting the discovery from Section \ref{sec:reg_up_dn_results} that optimal profile selection should be tailored to each hour. Our proposed algorithm outperforms the 50-50 strategy by almost 20\%, and it also outperforms the trivial approach of 0\% participation, supported by the fact that profits are positive.

\begin{figure}[htbp!]
\centering
\includegraphics[width=\columnwidth]{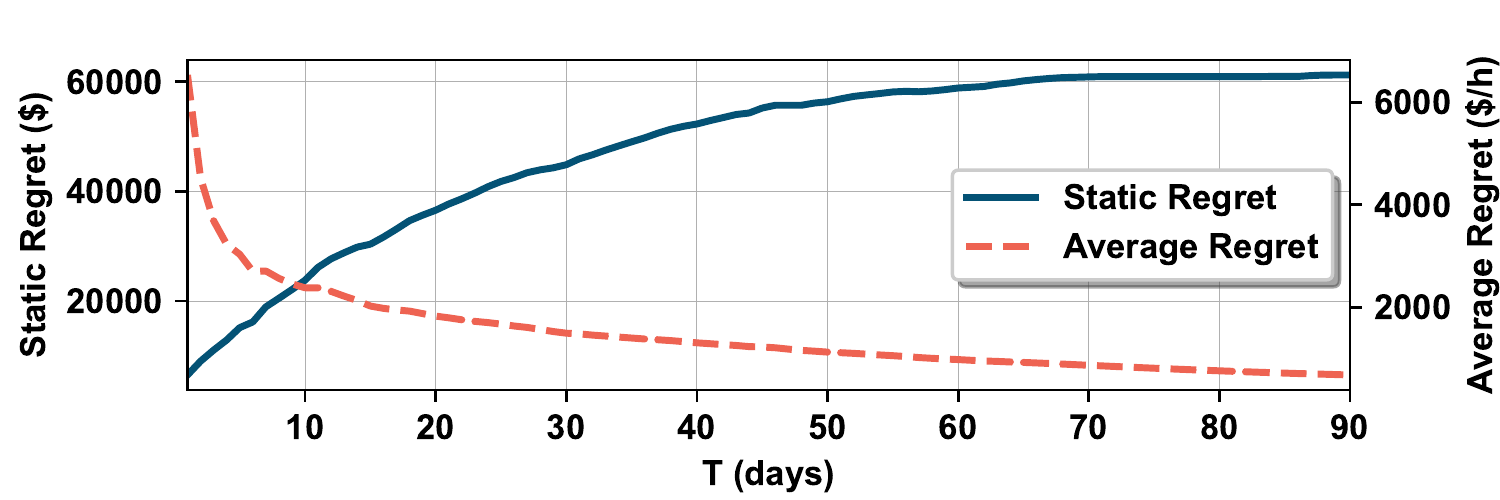}
\caption{Static and average regret over multiple rounds during a winter period, from Dec 2021 to Mar 2022.}
\label{fig:regret}
\end{figure}

\subsection{Online Optimization Case}
\label{sec:online_results}
In this section, we perform the online optimization method described in Section \ref{sec:online}, over the three-month 2021/2022 winter period and show how the algorithm learns to minimize average regret over time, as shown in Fig. \ref{fig:regret}. As expected, even though the static regret continues to increase, it eventually plateaus indicating closeness to the optimal profile. As the total number of rounds $T$ increases, the average regret approaches zero. The strategy can closely track the best solution on hindsight after 70 days without much prior knowledge.

\section{Conclusion}\label{sec:conclusion}

In this work, we presented a broad set of problem formulations in the context of cryptocurrency miners' participation in ancillary service programs, including online and offline optimization, as well as risk-aware portfolio selection. We evaluated our approach using numerical simulations based on real-world multi-year historical records of Texas-based electricity and ancillary services market data collected from ERCOT. We also rely on real-world data of Bitcoin prices, cryptomining machine specifications and cryptomining difficulty records to perform a realistic set of numerical experiments with results that can assist crptomining facilities in making informed decisions and maximizing their returns from ancillary service participation. We consider scenarios of cryptomining facilities optimizing over multiple ancillary service programs with multiple cryptomining machine types at their disposal, based on real-time pricing from both ERCOT and the Bitcoin network. Our results reveal the sensitivity of optimal decisions to hour-of-day, and further research can expand on this proposed formulation to consider more complex scenarios. The understanding of the optimal strategy that can be taken by crytominers will also benefit grid operators in their design of ancillary services, to reduce the operating cost and maintain system stability. 

%Future work will investigate the actual deployment of such strategy in an online setting in the real world. 

% \newpage

\bibliographystyle{ACM-Reference-Format}
\bibliography{ref}

\appendix
\section{Proof of Theorem~\ref{thm:reg}} \label{appendixA}
\begin{proof} During the proof of this Theorem we will be using the truncated exponential distribution. Here we present the following results for this distribution, which will be used in the proof.

\begin{equation} 
		PDF(\epsilon) \mbox{=} \begin{cases}
		 \frac{\lambda e^{-\lambda \epsilon}}{1-e^{-\lambda}}	\, & \mathrm{if } \quad  0\leq \epsilon \leq 1,\\
			0 &  \mathrm{otherwise} ,
		\end{cases} 
\end{equation} 
it can be shown that given the above truncated distribution, we can calculate its expected value as follows: 
\begin{equation} 
\EX[\epsilon] = \int_{0}^{1} \epsilon \frac{\lambda e^{-\lambda \epsilon}}{1-e^{-\lambda}} \,d\epsilon = \frac{1-(\lambda+1)e^{-\lambda}}{\lambda(1-e^{-\lambda)}},
\end{equation} 
and we also have:

\begin{equation} 
\int_{0}^{b} x \lambda e^{-\lambda x} \,dx = \big[\frac{x}{-e^{\lambda x}}\big]_{0}^{b} - \big[\frac{1}{\lambda} e^{-\lambda x}\big]_{0}^{b}
\end{equation} 

Now we start the proof. For the case with two types of mining devices, and $N$ programs, we have the following cost function:

\begin{align*}
    cost(\epsilon, c) :=
    \begin{cases}
        \sum_{i=1}^N c_i\left[r_1\epsilon_i - p_i\right] &
        \text{if} \quad \sum_{i=1}^N c_i\epsilon_i \leq c^M_1 \\
        \sum_{i=1}^N c_i\left[r_2\epsilon_i - p_i\right] &\\
        \hspace{4em} + c^M_1\left[r_1 - r_2\right] &
        \text{otherwise}
    \end{cases}   
\end{align*}

Hence, for the expected value of cost we have:
\begin{eqnarray} \label{expcost}
    &&\EX_{_{\mathcal{\epsilon}}}[cost(\epsilon, c)] = \sum_{i=1}^N c_i(r_1\EX [\epsilon_i] - p_i) + \notag\\
    &&(r_1 - r_2) \EX \big[ \min \{c^M_1 -\sum_{i=1}^N c_i\epsilon_i,0 \} \big]
\end{eqnarray}

Consider that we only have two programs of reg-up and reg-down with the following joint probability distribution function:

\begin{equation} \label{pdf}
		PDF(\epsilon_1, \epsilon_2) = 
		\theta *\delta(\epsilon_1)  \underbrace{\frac{\lambda_2 e^{-\lambda_2 \epsilon_2}}{1-e^{-\lambda_2}}}_{\textit{reg-down}} + (1-\theta) * \delta(\epsilon_2) \underbrace{ \frac{\lambda_1 e^{-\lambda_1 \epsilon_1}}{1-e^{-\lambda_1}} }_{\textit{reg-up}},
\end{equation} 
From the above distribution it can be seen that either reg-up or reg-down are being deployed. In other words, when $\epsilon_1=0$, reg-down is deployed, and when $\epsilon_2=0$, reg-up is deployed. Now we consider each of these cases individually, and combine them according to their probability. It should be noted that  $\epsilon_2$ is the deployment ratio of the reg-down program. Hence, in the cost function \eqref{expcost}, we replace it with $1-\epsilon_2$.

\textbf{Reg-down ($\epsilon_1=0$)}: First we consider a case where $\epsilon_1=0$. To determine the second term in the cost function \eqref{expcost}, we need to always check the following inequality:
\begin{equation}
    c_1\epsilon_1 + c_2 (1-\epsilon_2) \leq c^M_1
\end{equation}
In the case with $\epsilon_1=0$, it becomes $(1-\epsilon_2)c_2 \leq c_1^M$. We consider two cases. First, if $c_2 \leq c_1^M$, we always have $(1-\epsilon_2)c_2 \leq c_1^M$, which makes the second part of the cost function equal to zero. Hence for the cost we have:

\begin{eqnarray}
cost_{dn1} = -p_1c_1- p_2c_2 + c_2 r_1((\frac{-1+\lambda_2+e^{-\lambda_2}}{\lambda_2(1-e^{-\lambda_2})}) 
\end{eqnarray}

Second, if $c_2 > c_1^M$, the condition $(1-\epsilon_2)c_2 \geq c_1^M$ will be satisfied if we have $0 \leq \epsilon_2 \leq 1- \frac{c_1^M}{c_2}$. We calculate the second part of the cost function as follows:

\begin{eqnarray}
 &&E =\int_{0}^{1- \frac{c_1^M}{c_2}} (c_1^M- (1-\epsilon_2)c_2) \frac{\lambda_2 e^{-\lambda_2 \epsilon_2}}{1-e^{-\lambda_2}}  \,d\epsilon_2 \notag \\
&&  =  \frac{c_1^M-c_2}{1-e^{-\lambda_2}} \big[1-e^{-\lambda_2 (1-\frac{c_1^M}{c_2})} \big]+ \frac{c_2}{1-e^{-\lambda_2}} \big[\frac{c_2-c_1^M}{-c_2e^{\lambda_2 (1- \frac{c_1^M}{c_2})}} \notag \\
&& - \frac{1}{\lambda_2} e^{-\lambda_2 (1- \frac{c_1^M}{c_2})}+\frac{1}{\lambda_2} \big]= \frac{c_1^M-c_2}{1-e^{-\lambda_2}} + \notag \\
&&\frac{ c_2 }{\lambda_2(1-e^{-\lambda_2})}(1-e^{-\lambda_2 (1- \frac{c_1^M}{c_2})}) 
\end{eqnarray}
Hence, for the cost function in this case we have:

\begin{eqnarray}
&& cost_{dn2} = cost_{dn1}  +  \\
&& (r_1-r_2) \big[ \frac{c_1^M-c_2}{1-e^{-\lambda_2}} + \frac{ c_2 }{\lambda_2(1-e^{-\lambda_2})}(1-e^{-\lambda_2 (1- \frac{c_1^M}{c_2})}) \big] \notag
\end{eqnarray}
Now, in the following, we look into the second case when $\epsilon_2=0$.
\textbf{Reg-up ($\epsilon_2=0$)}:
When reg-up is deployed we have $\epsilon_2=0$. To solve this problem we consider three different cases for $c_1$ and $c_2$ as follows:

First, $c_1+c_2 \leq c_1^M$: in this case $\epsilon_1 c_1 + c_2 \leq c_1^M$, and $E=0$, so for the cost function we have:

\begin{eqnarray}
 cost_{up1} =c_1 \big(r_1 \frac{1-(\lambda_1+1)e^{-\lambda_1}}{\lambda_1(1-e^{-\lambda_1})}-p_1\big) +c_2 (r_1- p_2) 
\end{eqnarray}

Second, $c_1+c_2 > c_1^M$ and $c_2 < c_1^M $: in this case if we have $\epsilon_1 c_1+c_2 > c_1^M$ we have non-zero expected value, which happens when $\epsilon_1 > \frac{c_1^M-c_2}{c_1}$. Hence, we have:

\begin{eqnarray}
&& E =  \int_{\frac{c_1^M-c_2}{c_1}}^{1} (c_1^M- \epsilon_1 c_1-c_2) \frac{\lambda_1 e^{-\lambda_1 \epsilon_1}}{1-e^{-\lambda_1}}  \,d\epsilon_1 \\
&&  =  \frac{c_1^M-c_2}{1-e^{-\lambda_1}} \big[-e^{-\lambda_1}+e^{-\lambda_1 \frac{c_1^M-c_2}{c_1}} \big]- \notag \\
&&\frac{c_1}{1-e^{-\lambda_1}} \big[\frac{1}{-e^{\lambda_1}} -\frac{1}{\lambda_1} e^{-\lambda_1}+ 
\frac{c_1^M-c_2}{c_1e^{\lambda_1 \frac{c_1^M-c_2}{c_1}}}+ 
\frac{1}{\lambda_1} e^{-\lambda_1 \frac{c_1^M-c_2}{c_1}} \big]  \notag \\
&&= \frac{c_1+ c_2-c_1^M}{1-e^{-\lambda_1}} e^{-\lambda_1} +  \frac{c_1}{\lambda_1(1-e^{-\lambda_1})} ( e^{-\lambda_1}- e^{-\lambda_1 \frac{c_1^M-c_2}{c_1}}) \notag
\end{eqnarray}
Hence the cost function is:
\begin{eqnarray}
  && cost_{up2} = cost_{up1}  + (r_1-r_2) \big[ \frac{c_1+ c_2-c_1^M}{1-e^{-\lambda_1}} e^{-\lambda_1}\notag \\
  && \frac{c_1}{\lambda_1(1-e^{-\lambda_1})} ( e^{-\lambda_1}- e^{-\lambda_1 \frac{c_1^M-c_2}{c_1}}) \big] 
\end{eqnarray}

Third, $c_1+c_2 > c_1^M$ and $c_2 \geq c_1^M $: in this case we always have $\epsilon_1 c_1+c_2 \geq c_1^M$. Hence, we have:

\begin{eqnarray}
&& E =  \int_{0}^{1} (c_1^M- \epsilon_1 c_1-c_2) \frac{\lambda_1 e^{-\lambda_1 \epsilon_1}}{1-e^{-\lambda_1}}  \,d\epsilon_1 = \notag \\
&& c_1^M-c_2 -c_1 \big( \frac{1-(\lambda_1+1)e^{-\lambda_1}}{\lambda_1(1-e^{-\lambda_1})}  \big)
\end{eqnarray}
Hence, for the cost we have:
\begin{eqnarray}
   cost_{up3} =  cost_{up1} + (r_1-r_2) \big[ c_1^M-c_2-c_1 \big( \frac{1-(\lambda_1+1)e^{-\lambda_1}}{\lambda_1(1-e^{-\lambda_1})}  \big) \big]
\end{eqnarray}
From the probability distribution function \eqref{pdf} it can be seen that with probability $\theta$, reg-down is deployed, and with probability $1-\theta$, reg-up is deployed. Hence, the expected cost for $0 \leq c_1+c_2 \leq c_1^M+c_2^M$ can be calculated as
\begin{eqnarray}
    \text{cost}(t) =
    \begin{cases}
       Q_1 &
        \text{if} \quad c_1+c_2 < c_1^M \\
        Q_2 &
        \text{if} \quad c_2 \geq c^M_1 \\
        Q_3&
        \text{if} \quad otherwise
    \end{cases}
\end{eqnarray}
where, 

\begin{eqnarray}
   && Q_1 = \theta cost_{dn1}+ (1-\theta)cost_{up1} \notag \\
    &&  Q_2 = \theta cost_{dn2}+(1-\theta) cost_{up3}\notag \\
    && Q_3 =\theta cost_{dn1}+ (1-\theta)cost_{up2}
\end{eqnarray}
This completes the proof.
\end{proof}

\section{Proof of Theorem~\ref{online}}
\label{proof:online}
\begin{proof}
Let us denote $D$ as an upper bound
on the diameter of the feasible region $\mathcal{F}_c$:
\begin{equation}
    \forall x,y \in \mathcal{F}_c, \quad ||x-y|| \leq D
\end{equation}
 If the norms of the subgradients of the function $f$ over $\mathcal{F}_c$ are upper bounded by $G$, 
($||\nabla f(x)||\leq G, \, \, \forall x \in \mathcal{F}_c$), function $f$ is Lipschitz continuous with parameter $G$.
According to \cite{hazan2016introduction},
the online gradient descent algorithm with step sizes ${\eta_t = \frac{D}{G\sqrt{t}}, \, \, t \in [1:T] }$ guarantees the following for all $T \geq 1$:

\begin{equation}
    R \leq \frac{3}{2} GD \sqrt{T}
\end{equation}

To determine the value of $G$, and $D$ for our setting, we assume that  $ 0 \leq r_K(t) \leq r_{max}$, and $0 \leq p_i(t)<p_{max}$. To calculate $D$, one should note that the feasible region $\mathcal{F}_c$ is given by:

\begin{equation}
  \mathcal{F}_c \coloneqq \{c \in \mathbb{R}^N | \quad c \geq 0, \quad \sum_{i=1}^{N} c_i < C^M\},
\end{equation}

hence, the maximum distance of any two points in this regions is calculated as follows:
\begin{equation}
 D =\begin{cases}
        C^M&
        \text{if} \quad N=1 \\
         \sqrt{2}C^M &
        \text{if}\quad N \geq 2,
\end{cases}   
\end{equation}
On the other hand, to calculate $G$, we know that depending on $c(t)$, the subgradient of the cost function $cost_t(c(t))$ given in \eqref{costonline}, is always calculated as
$[r_{k_c}(t)\epsilon_i(t) - p_i(t)]_{i=1}^N$. Given that $r_{k_c}(t)\epsilon_i(t)$ and $p_i(t)$ are both positive, we have:
\begin{eqnarray}
      &&r_{k_c}(t)\epsilon_i(t) - p_i(t) \leq \max \{r_{k_c}(t)\epsilon_i(t), p_i(t) \} \leq \notag \\
      &&\max \{r_{k_c}(t), p_i(t) \} \leq max \{r_{max}, p_{max} \} 
\end{eqnarray}
Given that the subgradient is an $N$ dimensional vector, we have $G=\sqrt{N} \max \{r_{max}, p_{max}\}$. Hence, for $N\geq 2$, for the regret bound we have:

\begin{equation}
    R \leq \frac{3}{2} GD \sqrt{T} = \frac{3 }{2} * \max \{r_{max}, p_{max}\} * C^M \sqrt{2TN}
\end{equation}

This completes the proof.

\end{proof}

\end{document}